\documentclass{IEEEtran}
\usepackage{hyperref}
\usepackage{cite,mathtools}
\usepackage{amsmath,amssymb,array,amsfonts}
\newtheorem{remark}{Remark}
\newtheorem{definition}{Definition}
\newtheorem{theorem}{Theorem}
\newtheorem{example}{Example}

\title{Reduced Complexity Sum-Product Algorithm for Decoding Network Codes and In-Network Function Computation}
\begin{document}

\author{
\authorblockN{Anindya Gupta and B. Sundar Rajan\\}
\authorblockA{Department of Electrical Communication Engineering, Indian Institute of Science, Bangalore 560012, India\\ Email: \{anindya.g, bsrajan\}@ece.iisc.ernet.in}
}

\maketitle
\begin{abstract}
While the capacity, feasibility and methods to obtain codes for network coding problems are well studied, the decoding procedure and complexity have not garnered much attention. In this work, we pose the decoding problem at a sink node in a network as a marginalize a product function (MPF) problem over a Boolean semiring and use the sum-product (SP) algorithm on a suitably constructed factor graph to perform iterative decoding. We use \textit{traceback} to reduce the number of operations required for SP decoding at sink node with general demands and obtain the number of operations required for decoding using SP algorithm with and without traceback. For sinks demanding all messages, we define \textit{fast decodability} of a network code and identify a sufficient condition for the same. Next, we consider the in-network function computation problem wherein the sink nodes do not demand the source messages, but are only interested in computing a function of the messages. We present an MPF formulation for function computation at the sink nodes in this setting and use the SP algorithm to obtain the value of the demanded function. The proposed method can be used for both linear and nonlinear as well as scalar and vector codes for both decoding of messages in a network coding problem and computing linear and nonlinear functions in an in-network function computation problem.
\end{abstract} 

\begin{IEEEkeywords}
Network Coding, Decoding, Sum-Product Algorithm, Traceback, In-network Function Computation.
\end{IEEEkeywords}
\section{Introduction}
In contemporary communication networks, the nodes perform only routing, i.e., they copy the data on incoming links to the outgoing links. In order to transmit messages generated simultaneously from multiple sources to multiple sinks the network may need to be used multiple times. This limits the throughput of the network and increases the time delay too. Network coding is known to circumvent these problems \cite{Yeung}. In network coding intermediate nodes in a network are permitted to perform coding operations, i.e., encode data received on the incoming links and then transmit it on the outgoing links (each outgoing link can get differently encoded data), the throughput of the network increases. Thus, network coding subsumes routing. For example, consider the butterfly network \cite{Yeung} of Fig.~\ref{b_fly} wherein each link can carry one bit per link use, source node $S$ generates bits $b_1$ and $b_2$, and both sink nodes $T_1$ and $T_2$ demand both source bits. With routing only, two uses of link $V_3-V_4$ are required while with network coding only one. 
\begin{figure}[h]
\centering
\includegraphics[scale=0.45]{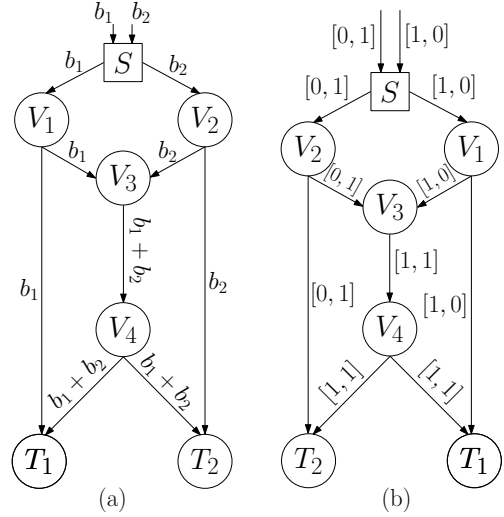}
\caption{The butterfly network: (a) A network code and (b) its global encoding vectors.}
\label{b_fly}
\vspace{-14pt}
\end{figure}

This is an example of single-source multi-sink linear multicast network coding, wherein there is a single source ($S$), generating a finite number of messages, ($x_1,x_2$), and multiple sinks, each demanding all the source messages and the encoding operations at all nodes are linear. In general, there may be several source nodes, each generating a different number of source messages, and several sink nodes, each demanding only a subset, and not necessarily all, of the source messages. Decoding at sink nodes with such general demands is studied in this paper. 

\subsection{Notations and Preliminaries}
We represent a network by a finite directed acyclic graph $\mathcal{N}=(\mathcal{V},\mathcal{E})$, where $\mathcal{V}$ is the set of vertices or nodes and $\mathcal{E} \subseteq \mathcal{V} \times \mathcal{V}$ is the set of directed links or edges between nodes. All links are assumed to be error-free. Let $F$ denote a $q$-ary finite field. The set $\{1,2,\ldots,n\}$ is denoted by $[n]$. The network has $J$ sources, $S_j,\,j\in [J]$, and $K$ sinks, $T_k,\,k\in [K]$. The source $S_j$ generates $\omega _j$ messages for all $j\in [J]$. Let $\omega=\sum _{j=1}^J \omega _j$ be the total number of source messages. The $\omega$-tuple of source messages is denoted by $x_{[\omega]}=(x_1,x_2,\ldots,x_\omega)$, where $x_i\in F$ for all $i\in [\omega]$. By $\mathbf{x}=(x_1,\ldots,x_{\omega})^T$ we denote the column vector of the source messages. The demand of the $k^{th}$ sink node is denoted by $D_k\subseteq [\omega]$. Given a set $I=\{i_1\ldots,i_l\}\subseteq [\omega]$, let $x_I=(x_{i_1},\ldots,x_{i_l})$, i.e., $x_{[\omega]}$ \emph{restricted} to $I$. For disjoint subsets $I$ and $J$ of $[\omega]$, we do not differentiate between $(x_I,x_J)$ and $x_{I\cup J}$. 
For a multi-variable binary-valued function $f(x_1,\ldots,x_\omega)$, the subset of $F^\omega$ whose elements are mapped to $1$ by $f(x_1,\ldots,x_\omega)$ is called its support and is denoted by  $\mathrm{supt}(f(x_{[\omega]}))$ and $\mathrm{supt}_{I}(f(x_{[\omega]}))$ denotes the $|I|$-tuples in the support restricted to $I$. A source message is denoted by edges without any originating node and terminating at a source node. Data on a link $e\in \mathcal{E}$ is denoted by $y_e$. 

A network code is a set of coding operations to be performed at each node such that the requisite source messages can be faithfully reproduced at the sink nodes. It can be specified using either local or global description \cite{Yeung}. The former specifies the data on a particular outgoing edge as a function of data on the incoming edges while the latter specifies the data on a particular outgoing edge as a function of source messages. Throughout the paper we use global description for our purposes. 
\begin{definition}
[Global Description of a network code \cite{Yeung}] An $\omega$-dimensional network code on an acyclic network over a field $F$ consists of $|E|$ global encoding maps $\tilde{f}_e : F^{\omega}\rightarrow F$ for all $e\in E$, i.e.,  $\tilde{f}_e ( \mathbf{x}) = y_e$.
\end{definition}
Let $e_i,i=1,\ldots,\omega$, be the incoming edges at the source, then $y_{e_i}=x_i$. 

When the intermediate nodes perform only linear encoding operations, the resulting network code is said to be a linear network code (LNC).
\begin{definition}
[Global Description of an LNC \cite{Yeung}] An $\omega$-dimensional LNC on an acyclic network over a field $F$ consists of $|E|$ number of $1\times\omega$ global encoding vectors $\mathbf{f}_e$ for all $e\in E$ such that $\mathbf{f}_e\cdot \mathbf{x}=y_e$.
\end{definition}
The global encoding vectors for the incoming edges at the source are standard basis vectors for the vector space $F^{\omega}$. The global encoding vectors of the LNC for butterfly network is given in Fig.~\ref{b_fly}(b). 

Hereafter we assume that the network is feasible, i.e., demands of all sink nodes can be met using network coding, and the global description of a network code (linear or nonlinear) is available at the sink nodes. If a sink node demands $\omega'$ ($\leqslant$ $\omega$) source messages, it will have at least $\omega'$ incoming edges. The decoding problem is to reproduce the desired source messages from the coded data received at the incoming edges. Thus, decoding amounts to solving a set of at least $\omega'$ simultaneous equations (linear or nonlinear) in $\omega$ unknowns for a specified set of $\omega'$ unknowns. Hence, the global description of a network code is more useful for decoding. 

While decoding of nonlinear network codes has not been studied, the common technique used for decoding an LNC for multicast networks is to perform Gaussian elimination \cite{Lun1,Lun2}, which requires $\mathcal{O}(\omega^3)$ operations, followed by backward substitution, which requires $\mathcal{O}(\omega^2)$ operations \cite{Num}. This is not recommendable when the number of equations (incoming coded messages) and/or variables (source messages) is very large. In such cases, iterative methods are used. Convergence and initial guess are some issues that arise while using iterative methods \cite{Strang}. 

We propose to use the sum-product (SP) algorithm to perform iterative decoding at the sinks. A similar scheme for decoding multicast network codes using factor graphs \cite{Frey} was studied in \cite{Salmond} in which the authors considered the case of LNCs. The problems associated with the proposed decoding scheme in \cite{Salmond} are:
\begin{itemize}
\item To construct the factor graph, full knowledge of network topology is assumed at the sinks which is impractical if the network topology changes. For a particular sink node (say $T$), the factor graph constructed will have $\omega +|E|$ variable nodes and $|E|+ |In(T)|$ factor nodes, where $In(T)$ is the set of incoming edges at node $T$.
\item Complete knowledge of local encoding matrix \cite{Yeung} of each node is assumed at the sinks which again is impractical since local encoding matrix for different nodes will have different dimensions and hence variable number of overhead bits will be required to communicate to downstream nodes which will incur huge overhead.
\end{itemize}
We also point out that the motivating examples, \textit{viz.}, Examples 1 and 4, given in \cite{Salmond} for which the proposed decoding method claims to exploit the network topology admits a simple routing solution and no network coding is required to achieve maximum throughput. Solving a system of linear equations in Boolean variables is also studied in \cite[Ch.~18]{Mez}.

\subsection{Contributions and Organization}
\begin{table*}[t] 
\caption{Number of semiring operations}
\label{tab_nops}
\centering
\small
\renewcommand{\tabcolsep}{10pt}
\renewcommand{\arraystretch}{1.8}
\begin{tabular}{|>{\centering\arraybackslash}m{1in}|>{\centering\arraybackslash}m{2.2in}|>{\centering\arraybackslash}m{2.9in}|}
\hline
 & \textbf{Sum-Product Algorithm} & \textbf{arg-Sum-Product Algorithm} \\
\hline 
\textbf{Single-vertex} & $\sum\limits_{z\in Z}d_zq_z - \sum\limits_{e\in E}q_e - \sum\limits_{v\in V}q_v$ & $\sum\limits_{z\in Z}d_zq_z - \sum\limits_{e\in E}q_e - \sum\limits_{v\in V}q_v +\, q_r\, -1$ \\[6pt]
\hline
\textbf{All-vertex} & $\sum\limits_{z\in Z}(4d_z-5)q_z \,+\, 2\sum\limits_{w\in W}q_w \,-\, 2\sum\limits_{e\in E}q_e$ & $\sum\limits_{z\in Z}(4d_z-5)q_z \,+\, 2\sum\limits_{w\in W}q_w \,-\, 2\sum\limits_{e\in E}q_e\,+\,\sum\limits_{z\in Z}q_z\,-\,|Z|$ \\[6pt]
\hline
\textbf{Single-vertex with Traceback} & Not Applicable & $\sum\limits_{z\in Z}d_zq_z - \sum\limits_{e\in E}q_e + \sum\limits_{w\in W}q_w \,-\, |Z|$ \\[6pt]
\hline
\end{tabular}
\begin{tabular}{p{6.8in}}
Note: $V$ is the set of variable nodes, $W$ is the set of factor nodes, $Z=V\cup W$, $E$ is the set of edges, $G=(Z,E)$ is the factor graph, $r$ is the chosen root node in $G$, and $a_r=1$ if $r\in W$ and $0$ otherwise. (See Section~IV for a complete discussion.)
\end{tabular}
\vspace{-6pt}
\end{table*}
The contributions and organization of this paper are as follows:
\begin{itemize}
\item In Section~III-A, we pose the problem of decoding of linear and nonlinear network codes as a \emph{marginalize a product function problem} (MPF) and construct a factor graph using the global description of network codes. For a particular sink node, the constructed graph will have fewer vertices than in \cite{Salmond} and hence the number of messages and operations performed will also be fewer. Unlike in \cite{Salmond}, our scheme requires only the knowledge of global encoding maps/vectors of incoming edges at a sink node and not the entire network structure and coding operation performed at each node. 
\item In Section~III-B, we utilize \textit{traceback} \cite{LP} instead of running the multiple-vertex version of the SP algorithm which results in reduction in the number of operations. Application and advantage of using traceback over multiple-vertex SP algorithm for decoding at sinks with general demand is demonstrated.
\item We discuss the utility and the computational complexity of the proposed technique in Section~IV. We give the number of semiring operations required to perform single- and all-vertex SP algorithm for a class of MPF problem where we are interested not in the marginal function at a particular vertex but in the values of the variables/arguments in the local domain of that vertex that causes that marginal function to attain certain value in the semiring, for example, maximum in a max-sum or max-product semiring or minimum in a min-sum or min-product semiring. We call such problems as arg-MPF problems and refer to the application of the SP algorithm to such problems as the arg-SP algorithm. We show that the number of semiring operations required in performing single-vertex arg-SP algorithm with traceback is strictly less than that of all-vertex SP algorithm (see Table~\ref{tab_nops} for a comparison). Hence, the decoding complexity of a network code using SP decoding with traceback is strictly less than that without using traceback. For sink nodes which demand all the source messages, the notion of \emph{fast decodable network codes} is defined and a sufficient condition for the same is identified. 
\item In Section~V, we consider the in-network function computation problem wherein the sink nodes demand a function of source messages. A network code for such a problem ensures computation of the value of the demanded function at a sink node given the coded messages on its incoming edges and not the reproduction of the values of the arguments of the function. Thus, multiple message vectors may evaluate to the same incoming coded messages and the demanded function value. In Section~V-B, we show that obtaining one such message vector can be posed as an MPF problem and that obtaining it suffices for computation of the demanded function. Subsequently, we give a method to construct a factor graph for each sink node and use the SP algorithm to solve the MPF problem.
\end{itemize}

We present a brief overview of the SP algorithm in Section~II. Preliminaries of in-network function computation are given in Section~V-A. We conclude the paper with a discussion on scope for further work in Section~VI.

\section{The Sum-Product Algorithm and Factor Graphs}
In this section, we review the computational problem called the MPF problem and specify how SP algorithm can be used to efficiently solve such problems. An equivalent method to efficiently solve MPF problems is given in \cite{Aji} and is called the \emph{generalized distributive law} (GDL) or the \emph{junction tree algorithm}. The simplest example of SP algorithm offering computational advantage is the distributive law on real numbers, $a\cdot (b+c)=a\cdot b+a\cdot c$; the left hand side of the equation requires fewer operation than the right hand side. Generalization of addition and multiplication is what is exploited by the SP (or the junction tree) algorithm in different MPF problems. The mathematical structure in which these operations are defined is known as the commutative semiring \cite{Aji}. 

\begin{definition}
A commutative semiring is a set $R$, together with two binary operations ``$+$'' (\textit{addition or sum}) and ``$\cdot$'' (\textit{multiplication or product}), which satisfy the following axioms:
\begin{enumerate}
\item The operation ``$+$'' satisfies closure, associative, and commutative properties; and there exists an element ``$0$'' (\textit{additive identity}) such that $r+0=r$ for all $r\in R$.
\item The operation ``$\cdot$'' satisfies closure, associative, and commutative properties; and there exists an element ``$1$'' (\textit{multiplicative identity}) such that $r\cdot 1=r$ for all $r\in R$.
\item The operation ``$\cdot$'' \textit{distributes} over ``$+$'', i.e., $r_1\cdot r_2 + r_1\cdot r_3 = r_1\cdot (r_2+r_3)$ for all $r_1,r_2,r_3\in R$
\end{enumerate}
\end{definition}

Different semirings are used for different MPF problem, each with a different notion of ``$+$'' and ``$\cdot$''. Some examples are listed below.
\begin{enumerate}
\item Application of the SP algorithm to the discrete Fourier transform yields the FFT algorithm; the semiring is the set of complex numbers with the usual addition and multiplication \cite{Frey,Aji}.
\item ML decoding of binary linear codes is also an MPF problem and application of SP algorithm yields the Gallager-Tanner-Wiberg decoding algorithm over a Tanner graph; the semiring is the set of positive real numbers with ``$\min$'' as sum and ``$+$'' as product, called the min-sum semiring \cite{Frey,Aji}. The BCJR algorithm for decoding turbo codes and the LDPC deocoding algorithm are some other applications of the SP algorithm.
\item Application to the ML sequence estimation, for instance in decoding convolutional codes, yields the Viterbi algorithm \cite{Aji}; the semiring is again the min-sum semiring.
\item Recently, the GDL has been shown to reduce the ML decoding complexity of space-time block codes in \cite{LP}; the semiring applicable is the min-sum semiring of complex number. The authors introduced traceback for GDL and used it to further lower the number of operations.
\end{enumerate}

Thus, the SP algorithm and the GDL subsume as special cases many well known algorithms. 

\subsection{MPF Problems in the Boolean Semiring}
The Boolean semiring is the set $\{0,1\}$ together with the usual Boolean operations $\vee$ (OR) and $\wedge$ (AND). We denote it by $R=(\{0,1\},\vee,\wedge)$. The elements $0$ and $1$ are the \emph{additive} and \emph{multiplicative identities} respectively. The MPF problem defined for this semiring is described below. Let $x_1,x_2,\ldots,x_n$ be a collection of variables taking values in finite alphabets $A_1,A_2,\ldots,A_n$, respectively. For $I=\{i_1,\ldots,i_k\}\subseteq [n]$, let $x_I=(x_{i_1},\ldots,x_{i_k})$ and $A_I=A_{i_1}\times \ldots \times A_{i_k}$. Let $\mathcal{S}=\{S_1,S_2,\ldots,S_M\}$ be a family of $M$ subsets of $[n]$ such that for each $j\in [M]$, there is a function $h_j:A_{S_j}\rightarrow R$. These functions are called the \emph{local functions}, the set of variables in $x_{S_j}$ is called the \emph{local domain} of $h_j$, and $A_{S_j}$ is the associated \emph{configuration space}. The \emph{global function}, $g: A_{[n]}\rightarrow R$ and the \emph{marginal function}, $g_I: A_I\rightarrow R$, associated with a subset $I$ of $[n]$ are defined as follows:
\begin{align} \nonumber
g(x_1,x_2,\ldots,x_n)=\bigwedge _{j=1}^M h_j(x_{S_j})
\end{align}
and
\begin{align} 
\label{eqn_marg}
g_I(x_I) = \qquad\bigvee _{\mathclap{x_{[n]\backslash I}\in A_{[n]\backslash I}}}\quad g(x_1,x_2,\ldots,x_n).
\end{align}

If we are interested in the support of the marginal $g_I(x_I)$, then the instance $x_I^*$, if unique, of $x_I$ for which $g_I(x_I^*)=1$ is obtained as follows:
\begin{align} \label{eqn_supt}
x_I^*=\mathrm{supt}\; g_I(x_I),
\end{align}
and for $J\subseteq I$, $x_J^*$ (if unique) can be obtained as follows:
\begin{align} \label{eqn_supt_sub}
x_J^*=\underset{J}{\mathrm{supt}}\; g_I(x_I).
\end{align}
If the instances of $x_I$ obtained using \eqref{eqn_supt} are not unique, i.e., the support contains more than one $I$-tuple for which $g_I$ evaluate to $1$, then all these can be collected in a set, say, $B_I$, where $B_I=\mathrm{supt}\; g_I(x_I)\subseteq A_I$. 
%
This is the arg-MPF problem for the Boolean semiring where an instance of some variables (arg) that causes a marginal function (MPF) to evaluate to $1$ is required. Other examples of arg-MPF problems include decoding of classical error-correcting codes, ML sequence detection using Viterbi algorithm, and ML decoding of space-time block codes in appropriate min-sum semirings; in all these problem we are interested in obtaining the instance of the variables that cause the marginal functions to evaluate to an element in the semiring which is the least when compared to evaluations at other possible instances of the variables. Similarly, over the Boolean semiring, the instances of a subset of variables which causes some marginal function to take value $0$ or $1$ may be of interest. When evaluation to $1$ is required, we use \eqref{eqn_supt} to obtain such instance(s).

\begin{remark}\label{rem_supt_or}
For a binary-valued function $f$ of $n$ variables $x_1,\ldots,x_n$ such that $x_i\in A_i$ for all $i\in [n]$, where $A_i$s are finite alphabets, finding a vector in its support, $\mathrm{supt}\;f(x_{[n]})$, and outputting OR of all the values it takes for different instances of input variables, $\bigvee_{x_{[n]}\in A_{[n]}}\;f(x_{[n]})$, can be seen as the same operation with different outputs. If a table of values of $f$ for different instances of $x_{[n]}$ is given, both go through the values in the columns of function values, and when a $1$ is encountered for the first time, the former outputs the instance of input variables and the latter outputs $1$. Thus, both these operations require same number of comparison which is at most $A_{[n]}-1$.
\end{remark}
\begin{remark}\label{rem_supt_ops}
The number of comparisons required in \eqref{eqn_supt_sub} is independent of $J$ and is at most $|A_I|-1$.
\end{remark}

The Boolean satisfiability problem is an example of the MPF problem over the Boolean semiring \cite{Frey}. Given a set of $M$ Boolean expressions in $n$ variables, a Boolean satisfiability problem asks whether there exists an assignment of $0$ or $1$ to the variables such that all the expressions evaluate to $1$ simultaneously. For example, let $h_1=x_1\vee x_2$ and $h_2=x_2\wedge(x_3\vee x_4)$ be two Boolean expressions in $4$ variables and the objective is to determine whether out of $16$ $(=2^4)$ possible values of $(x_1,x_2,x_3,x_4)$ there exists one for which both $h_1$ and $h_2$ evaluate to $1$. 
\begin{align} \nonumber
h(x_1,x_2,x_3,x_4)=\qquad\bigvee _{\mathclap{(x_1,x_2,x_3,x_4)\in \{0,1\}^4}} \;\; h_1(x_1,x_2)\wedge h_2(x_2,x_3,x_4).
\end{align}
The function $h$ evaluates to $1$ if there exists one such assignment and $0$ otherwise. The function $h$ can be taken as the marginal function of the global function $h_1(x_1,x_2)\wedge h_2(x_2,x_3,x_4)$ associated with the set $\{x_1,x_2,x_3,x_4\}$. The assignment, if unique, that satisfies all the Boolean expressions is $\mathrm{supt}\; h(x_1,x_2,x_3,x_4)$. If multiple assignments satisfy the expressions, then they can be collected in a set as stated before. These two cases, unique and non-unique solutions, will arise in decoding network codes and in-network function computation problems respectively.

Solving a system of $M$ linear or polynomial equations in $n$ variables, say $p_1(x_{[n]})=c_1, p_2(x_{[n]})=c_2, \ldots,p_M(x_{[n]})=c_M$ over a finite field, where $c_1,\ldots,c_M$ are constants, can also be posed as an arg-MPF problem over the Boolean semiring as follows:
\begin{align*} 
x_{[n]}^*&=\mathrm{supt}\; \bigwedge _{i=1}^M \delta(p_i(x_{[n]}),c_i),
\end{align*}
where $\delta$ is the Kronecker delta function defined as
\[
 \delta(a,b) =
  \begin{cases}
   0, & \text{if}\; a \neq b \\
   1, & \text{if}\; a = b.
  \end{cases}
\]
The local functions are $\delta(p_1(x_{[n]}),c_1),\ldots,\delta(p_M(x_{[n]}),c_M)$.

Compared to the MPF problems, the arg-MPF problems requires an additional step of obtaining the desired support set. 

\subsection{The SP Algorithm}
The SP algorithm is an efficient way of computing the marginal functions \eqref{eqn_marg}, which may require $\mathcal{O}(A_{[n]})$ operations if computed in a brute-force manner. It involves iteratively passing \emph{messages} along the edges of a \emph{factor graph}, $G=(V\cup W,E)$, associated with the given MPF problem. Let $Z=V\cup W$. The factor graph is a bipartite graph. Vertices in $V$ are called variable nodes; one for each variable $x_i$ for all $i\in [n]$ ($|V|=n$) and are labeled $x_i$. The local domain and configuration space associated with a variable node with label $x_i$ are $\{x_i\}$ and $A_i$ respectively. A variable node does not have a local function. The vertices in $W$ are called the factor nodes; one for each local function $h_j(x_{S_j})$ for all $j\in [M]$ ($|W|=M$) and are labeled $h_j$. For a factor node with label $h_j$, its local kernel is $h_j(x_{S_j})$, local domain is $x_{S_j}$, i.e., the set of variables which are its arguments, and the configuration space is $A_{S_j}$. A variable node $x_i$ is connected to a factor node $h_j$ iff $x_i$ is an argument of $h_j$, i.e., $i\in S_j$. 

Let $N(x_i)$ denote the set of factor nodes adjacent to the variable node $x_i$, i.e., set of local functions with $x_i$ as an argument, and $N(h_j)$ denote the set of variable nodes adjacent to the factor node $h_j$, i.e., the local domain $x_{S_j}$ of $h_j$. The directed message passed from a variable node $x_i$ to an adjacent factor node $h_j$ and vice versa are as follows:
\begin{align}
\label{eqn_msg1}
\mu _{x_i\rightarrow h_j}(x_i)= \bigwedge _{h'\in N(x_i)\backslash h_j}\mu _{h'\rightarrow x_i} (x_i)
\end{align}
\begin{align}
\label{eqn_msg2}
\mu _{h_j \rightarrow x_i}(x_i)= \bigvee _{x_{S_j\backslash i}\in A_{S_j\backslash i}} h_j(x_{S_j}) \bigwedge _{x'\in N(h_j)\backslash x_i}\mu _{x'\rightarrow h_j}(x')
\end{align}
The messages are actually tables of values containing value of the messages corresponding to different values of their arguments. 

Depending on the requirement, we may need to evaluate marginal(s) at only one, a few or all nodes in the factor graph; the versions of SP algorithm applied to these cases are referred to as the single-vertex, multiple-vertex, and all-vertex SP algorithm respectively. In all these cases, all the messages are initially directed to one node, called the root, i.e., all the edges are directed towards the root and the messages are passed along the direction of the edge. The algorithm starts at the leaf nodes (nodes with degree one) with these nodes passing messages to the adjacent nodes. If a leaf node is a variable node, then the message value is $1$ (the semiring multiplicative identity) for all possible values the variable can take, i.e. $\mu _{x_i\rightarrow h_j}(x_i)=1$ for all $x_i\in A_i$, where $h_j(x_i)$ is the unique local function with $x_i$ as an argument. If a leaf node is a factor node, then its local domain will contain only one variable, say $x_i$, and $\mu _{h_j \rightarrow x_i}(x_i)=h_j(x_i)$ for all $x_i\in A_i$. Once a vertex has received messages from all but one of its neighbors, it computes its own message using \eqref{eqn_msg1} or \eqref{eqn_msg2}, and passes it to the neighbor from which it has not yet received the message. This continues until the root has received messages on all its edges. Now the root computes and passes the messages to its neighbors and the process continues with messages being passed on each edge in the opposite direction, i.e., away from the root. This message passing terminates when all the nodes at which marginals are required to be computed have received messages from all its neighbors. After receiving messages from all its neighbors, a variable node $x_i$ computes its marginal function $g_i$ as follows:
\begin{align*} 
g_i(x_i) = \bigwedge _{h' \in N(x_i)}\mu_{h' \rightarrow x_i}(x_i),
\end{align*}
and the value $x_i^*$ for which $g_i(x_i)=1$ is 
\begin{align*} 
x_i^*=\mathrm{supt}\; g_i(x_i).
\end{align*}
The marginal function at a factor node $h_j$ can be computed as follows:
\begin{align*} 
g_j(x_{S_j}) = h_j(x_{S_j})\bigwedge _{x' \in N(h_j)}\mu_{x' \rightarrow h_j}(x'),
\end{align*}
and required supports can be computed using \eqref{eqn_supt} or \eqref{eqn_supt_sub}. As stated in Section~II-A, if there are multiple instances of an argument for which a marginal function evaluate to $1$, then they can be collected in a set.

To obtain the correct value of the required marginal functions, it is essential that the factor graph be free of cycles \cite{Frey}. If there are cycles, these may not be the correct values and the sets $B_i$ and $C_j$ may contain some undesired instances of arguments for which the marginals take value $0$, in addition to the support of the marginals. We use \textit{variable stretching} (refer to \cite[Sec.~VI-B and C]{Frey} for a detailed description) to eliminate cycles; this is explained below. Let $G$ be a connected factor graph with cycles, $N(x)$ be the neighbors of a variable node $x$ in $G$, and let $T$ be a spanning tree of $G$. Every variable node $x$ is connected to all the factor nodes in $N(x)$ in $G$ but not in $T$. In $T$, there is a unique path from every variable node $x$ to the factor nodes in $N(x)$, since it is a tree. For each variable $x$, add $x$ to the local domains of all the nodes in the aforementioned unique paths; this is referred to as stretching variable $x$. The resulting factor graph with enlarged local domains is acyclic and is denoted by $G'$. The SP algorithm applied to $G'$ will give the exact marginal functions \cite[Sec.~VI-C]{Frey}. If the factor graph is not connected, then we find a spanning tree of each connected component and perform variable stretching in each of the trees. This method is exemplified below.
\begin{figure}[h]
\centering
\includegraphics[scale=0.45]{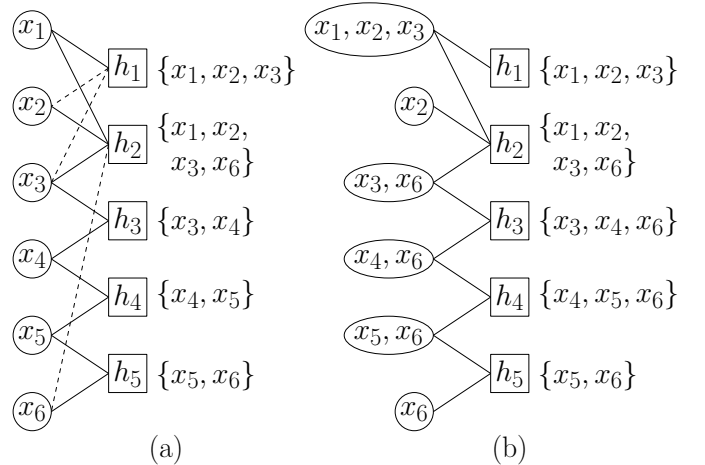}
\caption{Variable stretching (a) A factor graph with cycles and (b) its acyclic version after stretching variables $x_2$, $x_3$, and $x_6$ along the unique path from $x_2$ to $h_1$, $x_3$ to $h_1$, and $x_6$ to $h_2$ respectively in the factor tree obtained by removing dashed edges in (a).}
\label{vert_str}
\end{figure}

A factor graph with cycles is given in Fig.~\ref{vert_str}(a) and a spanning tree is obtained by deleting the dashed edges. In Fig.~\ref{vert_str}(a), $N(x_2)=\{h_1,h_2\}$, but $x_2$ is not connected to $h_1$ in the spanning tree and hence is stretched along the unique path $x_2-h_2-x_1-h_1$ in Fig.~\ref{vert_str}(a) between $x_2$ and $h_1$ resulting in the addition of $x_2$ to the local domain of variable node $x_1$ (local domain of $h_2$ already contains $x_2$). Similarly, $x_3$ is added to the local domain of variable node $x_1$ which lies in the unique path $x_3-h_2-x_1-h_1$ between $x_3$ and $h_1$, and $x_6$ is added to the local domains of $x_5$, $h_4$, $x_4$, $h_3$, and $x_3$ which lie on the unique path $x_6-h_5-x_5-h_4-x_4-h_3-x_3-h_2$ from $x_6$ to $h_2$. The resulting factor graph is depicted in Fig.~\ref{vert_str}(b).

The local functions of the factor nodes remain the same and the variable nodes are now labeled by the new enlarged local domains. In the new acyclic factor graph, we denote the variable and factor nodes by $v$ and $w$ and their local domains by $S_v$ and $S_w$ respectively. The SP algorithm on the modified graph proceed as before; message passing starts at the leaf node and terminates when each node has received a message from all its neighbors. The message passed from a variable node $v$ to a factor node $w$ in the new graph is 
\begin{align} \label{eqn_msg_mod1}
\mu _{v\rightarrow w}(x_{S_v\cap S_w})=\quad\bigvee _{\mathclap{x_{S_v\backslash S_w}\in A_{S_v\backslash S_w}}}\qquad\qquad\,\bigwedge _{\mathclap{w'\in N(v)\backslash w}}\,\mu _{w'\rightarrow v} (x_{S_{w'}\cap S_v}),
\end{align}
and that passed from a factor node $w$ to a variable node $v$ is
\begin{align} \label{eqn_msg_mod2}
\mu _{w\rightarrow v}(x_{S_w\cap S_v})=\;\;\bigvee _{\mathclap{x_{S_w\backslash S_v}\in A_{S_w\backslash S_v}}}\,h _w(x_{S_w})\;\;\bigwedge _{\mathclap{v'\in N(w)\backslash v}}\,\mu _{v'\rightarrow w} (x_{S_{v'}\cap S_w}),
\end{align}
where $h_w$ is the local function of factor node $w$. The marginal function of a variable node $v$ is
\begin{align} \label{eqn_state_mod1}
g_v(x_{S_v}) = \bigwedge _{w' \in N(v)}\mu_{w'\rightarrow v}(x_{S_{w'}\cap S_v}),
\end{align}
and that of a factor node $w$ is 
\begin{align} \label{eqn_state_mod2}
g_w(x_{S_w}) = h_w(x_{S_w})\bigwedge _{x' \in N(h_w)}\mu_{x' \rightarrow h_w}(x').
\end{align}
As before, the required supports can be computed using \eqref{eqn_supt} or \eqref{eqn_supt_sub}. 

From \eqref{eqn_msg_mod1}-\eqref{eqn_state_mod2}, it can be inferred that the number of operations required to compute messages and marginal functions in the SP algorithm will be $\mathcal{O}(A_{z^*})$, where $z^*$ is the node with the largest configuration space $A_{z^*}$.

\section{Decoding Network Codes Using The SP Algorithm}
In this section, we show that decoding a network code is an arg-MPF problem over the Boolean semiring. We provide a method to construct factor graph for decoding at a sink node using the SP algorithm. 

Though the factor graph approach \cite{Frey} and the junction tree approach \cite{Aji} are equivalent formulations to solve MPF problems, we prefer the former because of the difference in the amount of preprocessing required to obtain a junction tree as argued below:
\begin{enumerate}
\item The construction of a junction tree for an MPF problem requires \cite[Sec.~IV]{Aji}: \textit{(a)} construction of a \textit{local domain graph} with weighted edges, \textit{(b)} finding a maximum weight spanning tree, \textit{(c)} checking whether the sum of edges weights of the obtained maximum weight spanning tree is equal to $\sum_{j=1}^M|S_i|-n$, if yes then this tree is a junction tree for the MPF problem, otherwise we proceed with \textit{(d)} construction of a \textit{moral graph}, \textit{(e)} obtaining its \textit{minimum complexity triangulation} if it is not already triangulated, \textit{(f)} construction of the \textit{clique graph} of the triangulated moral graph, and \textit{(g)} finding a spanning tree which leads to minimum computational cost. To the nodes of this clique tree, called core in \cite{LP}, the local functions and variables of the MPF problem are attached \cite{Aji} to obtain the junction tree (a local function or a variable node is attached to a node of the core iff its local domain is a subset of the local domain of the said core node). Thus, the GDL always gives the exact solution of the MPF problems.
\item A factor graph is described by the local functions associated with the MPF problem. If it is acyclic, then the SP algorithm gives the exact solution, if not, it gives an approximate solution \cite{Frey}. The SP algorithm is known to perform well even if the factor graph has cycles, for example, in the iterative decoding of LDPC and turbo codes. As explained and exemplified in Section~II-B, cycles in a factor graph can be eliminated by first obtaining a spanning tree of the factor graph with cycles and then performing variable stretching \cite[Sec.~VI-C]{Frey}. The SP algorithm applied to the new acyclic factor graph will yield the exact marginal functions.
\end{enumerate}

\subsection{Network Code Decoding as an MPF Problem}
Given an acyclic network $\mathcal{N=(V,E)}$, the demands at each sink, $D_k,\,k\in [K]$ and a set of global encoding maps, $\{\tilde{f}_e:e\in \mathcal{E}\}$, that satisfy all the sink demands, the objectives at a sink, say $k^{th}$, is to find the instance of $x_{D_k}$ that was generated by the source(s) using the data it receives on its incoming edges, i.e.,
\begin{align}\label{eqn_mpf}
x_{D_k}^* = \underset{D_k}{\mathrm{supt}}\;\;\;\bigwedge _{\mathclap{{e \in In(T_k)}}}\; \delta \left( \tilde{f}_e (x_{[\omega]})\,,\,y_e \right)=\;\underset{D_k}{\mathrm{supt}}\,{g^{(k)}(x_{[\omega]})}. 
\end{align}
Here $g^{(k)}$ is the global function of the MPF problem at the $k^{th}$ sink. For an LNC, $\tilde{f}_e(x_{[\omega]})=\mathbf{f}_e\cdot\mathbf{x}$.

Thus, decoding a network code has the form of an arg-MPF problem over the Boolean semiring wherein we are interested only in some coordinates (specified by $D_k$) of the $\omega$-tuples in the support set and not the value of the global function. 

Since the solution $x_{D_k}^*$ is unique, individual coordinates $j\in D_k$ can be separately computed, i.e., 
\begin{align}
\label{eqn_marg1}
\begin{aligned} 
x_j^* = \underset{j}{\mathrm{supt}} \bigvee _{x_j\in F}\;g_j^{(k)}(x_j) \\ 
g_j^{(k)}(x_j) = \bigvee _{x_{[\omega]\backslash j}\in F^{\omega -1}}g^{(k)}(x_{[\omega])},
\end{aligned}
\end{align}
where $g_j^{(k)}(x_j)$ is a marginal function of the global function $g^{(k)}$.

The factor graph for decoding at sink $T_k,k\in [K]$ is constructed as follows:
\begin{enumerate}
\item Install $\omega$ \textit{variable nodes}, one for each source message. These vertices are labeled by their corresponding source messages, $x_i$. The local kernel of these nodes are 1.
\item Install $|In(T_k)|$ \textit{factor nodes} and label them $\tilde{f}_e, e\in In(T_k)$. The associated local domain of each such vertex is the subset $S_e\subseteq \{x_{[\omega]}\}$ of source messages that participate in that encoding map and the local function is $\delta(\tilde{f}_e(x_{S_e})\, , \, y_e)$.
\item A variable node is connected to a factor node iff the source message corresponding to that variable node participates in the encoding map corresponding to the said factor node.
\end{enumerate}

General form of a factor graph and the same for the two sink nodes of the butterfly network are given in Fig.~\ref{junc_grph}.
\begin{figure}[h]
\centering
\includegraphics[scale=0.5]{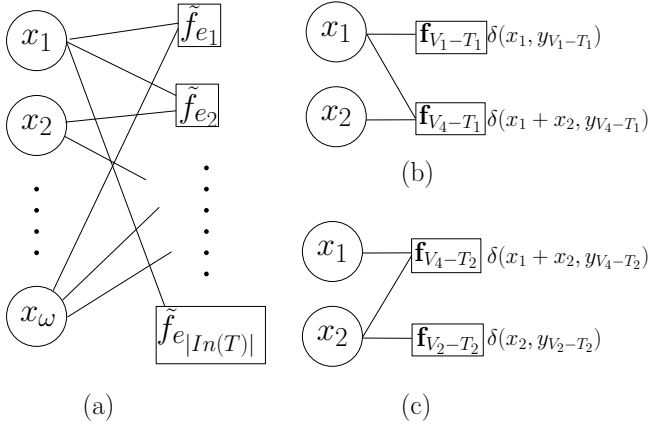}
\caption{(a) General form of a factor graph. (b) Factor graphs for $T_1$ and (c) $T_2$ of the butterfly network. Local function are given adjacent to the factor nodes.}
\label{junc_grph}
\vspace{-7pt}
\end{figure}

As specified in Section~II-B, the SP algorithm yields the correct value of the source messages if the factor graph is a tree. If not, then the cycles in the factor graph are eliminated via variable stretching on a spanning tree of the factor graph. Messages are computed using \eqref{eqn_msg_mod1} and \eqref{eqn_msg_mod2}, marginals using \eqref{eqn_state_mod1} and \eqref{eqn_state_mod2}, and the desired supports using \eqref{eqn_supt} or \eqref{eqn_supt_sub}.


\subsection{Traceback}
Since decoding network codes is an arg-MPF problem and not an MPF problem, we can use traceback \cite{LP} to reduce the number of operations.

We first demonstrate how traceback is used for decoding at a sink which demands all the source messages. If there exists a vertex whose local domain is the entire message set, then all the messages can be obtained by running single-vertex SP algorithm with this node as the root. If not, then assume that the single-vertex SP algorithm is run with a vertex, say $r$, as the root and the values $x_{S_r}^*$, $S_r\subset [\omega]$, of some source messages have been ascertained. Now, partition the local domain of a neighboring node $z$, as $x_{S_z}=x_A \cup x_B$, where $A=S_z\backslash S_r$ and $B=S_z\cap S_r$. Since $x_{S_r}^*$ is known, the value $x_B^*$ of $x_B$ for which $g_z(x_A,x_B)=1$ is also known. Then $x^*_A$ can then be obtained as follows:
\begin{align} \label{eqn_trcbk_mc}
\nonumber x^*_A & = \mathrm{supt}\;g_z(x_{S_z})=\mathrm{supt}\;g_z(x_A,x_B^*)\\
\nonumber & = \mathrm{supt}\; \mu_{r\rightarrow z}(x_B^*)\; \lambda_z(x_A,x_B^*)\\
 & = \mathrm{supt}\; \lambda_z(x_A,x_B^*),
\end{align}
where
\begin{align*}
\lambda_z(x_A,x_B)=h_z(x_{S_z})\bigwedge_{u\in N(z)\backslash r} \mu_{u\rightarrow z}(x_{S_u\cap S_z})
\end{align*}
is the partial marginal computed at $z$ while passing the message $\mu_{z \rightarrow r}(x_B)$ to the root $r$; the two are related as follows:
\begin{align*}
\mu_{z \rightarrow r}(x_B)&=\bigvee_{x_A}\;h_z(x_{S_z})\bigwedge_{u\in N(z)\backslash r} \mu_{u\rightarrow z}(x_{S_u\cap S_z})\\&=\bigvee_{x_A\in F^{|S_z|}}\!\lambda_z(x_A,x_B),
\end{align*}
where $h_z(x_{S_z})$ is the local function if $z$ is a factor node and assumed to be $1$ for all $x_{S_z}\in F^{|S_z|}$ if $z$ is a variable node. Thus, with traceback, neither computation of $\mu_{r\rightarrow z}(x_{S_r\cap S_z})$ nor that of marginal function $g_z(x_{S_z})$ is needed for any neighbor $z$ of $r$. 

The traceback step is performed until the values of all the source messages are obtained. This is done by obtaining source message values at a chosen root node $r$, followed by traceback on its neighbors, then the neighbors of neighbors of $r$, and so on. This can lead to considerable reduction in number of operations and is exemplified in Section III-C, Example~2.

We now present how traceback is performed in decoding network codes at a sink which demands only a subset of the source messages. Let $D\subset [n]$ denote the demand of a sink node. If there exists a vertex with local domain same as or containing $D$, then all the desired source messages can be obtained by running single-vertex SP algorithm with this node as the root. If not, then assume that the single-vertex SP algorithm is run with a vertex, say $r$, as the root. Let $z$ be a neighbor of $r$ with local domain $x_{S_z}$ partitioned as given in Fig.~\ref{fig_trcbk_bc}.
\begin{figure}[h]
\centering
\includegraphics[scale=0.5]{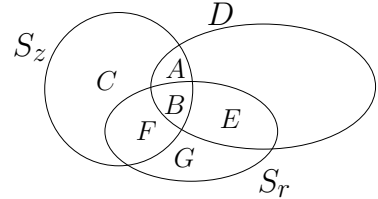}
\caption{Venn diagram of the demand set $D$ and local domains of $z$ and $r$.}
\label{fig_trcbk_bc}
\vspace{-7pt}
\end{figure}

Once $r$ has received all the messages, the marginal function $g_r(x_{S_r})$ is computed as follows:
\begin{align}\label{bc_marg}
g_r(x_{S_r})=h_r(x_{S_r})\mu_{z\rightarrow r}(x_B,x_F)\bigwedge_{\mathclap{u\in N(r)\backslash z}}\mu_{u\rightarrow r}(x_{S_u\cap S_r}).
\end{align}
Since the network code ensures decoding of only $(x_B,x_E)$, and not $(x_F,x_G)$, there exists a unique instance of $(x_B,x_E)$, denoted by $(x_B^*,x_E^*)$, and multiple instances of $(x_F,x_G)$, one of which is denoted by $(\widehat{x}_F,\widehat{x}_G)$, such that $g_r(x_B^*,x_E^*,\widehat{x}_F,\widehat{x}_G)=1$, i.e.,
\begin{align*}
(x_B^*,x_E^*,\widehat{x}_F,\widehat{x}_G)=\mathrm{supt}\;g_r(x_B,x_E,x_F,x_G).
\end{align*}
By Remark~\ref{rem_supt_ops}, computation of $(x_B^*,x_E^*,\widehat{x}_F,\widehat{x}_G)$ incurs no additional operations over computation of only $(x_B^*,x_E^*)$; both require at most $q^{|S_r|}-1$ comparisons. The message to be passed form $r$ to $z$ is 
\begin{align}\label{tb_bc}
\mu_{r\rightarrow z}(x_B,x_F)&=\bigvee_{\mathclap{(x_E,x_G)}}\;\;h_r(x_{S_r})\;\; \bigwedge_{\mathclap{u\in N(r)\backslash z}}\;\;\mu_{u\rightarrow r}(x_{S_u\cap S_r}).
\end{align}
Since $g_r(x_B^*,x_E^*,\widehat{x}_F,\widehat{x}_G)=1$, from \eqref{bc_marg} we have that 
\begin{align*}
h_r(x_{S_r})\bigwedge_{u\in N(r)\backslash z}\mu_{u\rightarrow r}(x_{S_u\cap S_r})=1
\end{align*}
for $x_{S_r}=(x_B^*,x_E^*,\widehat{x}_F,\widehat{x}_G)$, and consequently from \eqref{tb_bc}, $\mu_{r\rightarrow z}(x_B^*,\widehat{x}_F)=1$.

Now at $z$, $(x_A^*,\widehat{x}_C)$ is computed as follows:
\begin{align*}
(x_A^*,\widehat{x}_C)&=\mathrm{supt}\;g_z(x_A,x_B^*,x_C,\widehat{x}_F)\\ 
&=\mathrm{supt}\;\mu_{r\rightarrow z}(x_B^*,\widehat{x}_F)\;\lambda_z(x_A,x_B^*,x_C,\widehat{x}_F)\\
&=\mathrm{supt}\;\lambda_z(x_A,x_B^*,x_C,\widehat{x}_F),
\end{align*}
where $\mu_{r\rightarrow z}(x_B^*,\widehat{x}_F)=1$ as argued above and $\lambda_z(x_{S_z})$ is the partial message computed at $z$ while passing message $\mu_{z\rightarrow r}(x_B,x_F)$ to $r$; the two are related as follows:
\begin{align*}
\mu_{z\rightarrow r}(x_B,x_F)&=\bigvee_{(x_A,x_C)}h_z(x_{S_z})\bigwedge_{u\in N(z)\backslash r}\mu_{u\rightarrow z}(x_{S_u\cap S_z})\\&=\bigvee_{(x_A,x_C)}\lambda_z(x_{S_z}).
\end{align*} 
As before, by Remark~\ref{rem_supt_ops}, computation of $(x_A^*,\widehat{x}_C)$ incurs no additional operations over computation of only $x_A^*$; both require at most $q^{|A|+|C|}-1$ comparisons. 

This process is repeated on other neighbors of $r$, followed by neighbors of neighbors of $r$, and so on until values of all the messages in the demand set have been determined.

Thus, for a sink with a general demand set $D\subseteq [\omega]$, in the single-vertex SP algorithm with traceback, first the single-vertex SP algorithm is used with some node (preferably one whose local domain includes some of the demanded messages) as the root wherein all messages are directed towards it, its marginal function is computed, and then the support of the marginal. In the traceback step, appropriate supports of partial marginals, which were already computed while passing the messages towards the root, of some more nodes are computed; this involves only comparison operations. Let the root together with the set of nodes involved in the traceback step, i.e., nodes whose union of local domains contain the demand set, be denoted by $Z'$. 

In the multiple-vertex SP algorithm, computing marginal functions of nodes in $Z'$ is enough since their appropriate support will satisfy the sink's demands. As above, first the single-vertex SP algorithm is used. When the root has received all the messages, it passes messages to its neighbors and message passing continues until all other nodes in $Z'$ have received all the messages. After this, marginal functions of the nodes in $Z'$ is computed, and then appropriate support (only the intersection of local domain of the root and $D$) of the marginal function is computed. 

By Remark~\ref{rem_supt_ops}, computation of supports in both the single-vertex with traceback and in the last step of the multiple-vertex SP algorithm incurs the same computational cost. But the latter involves computation of additional messages (directed away from root) and marginals (of nodes in $Z'$ other than the root). Hence, the traceback step reduces the number of operations required in decoding a network code. We refer to the use of single-vertex SP decoding followed by traceback as the \emph{reduced complexity SP decoding}. Exact number of operations required in the all-vertex SP algorithm and single-vertex SP algorithm with traceback is derived in Section~IV.


\subsection{Illustrations}
We now present some examples illustrating use of the SP algorithm to decode network codes.

\begin{example}
Consider the butterfly network of Fig.\ref{b_fly}. Here $q=\omega =2$. The factor graphs for two sink nodes are given in Fig.~\ref{junc_grph}(b) and (c). The messages passed and state computations for decoding at $T_1$ are as follows:
\begin{align*} 
\mu_{x_2\rightarrow \mathbf{f}_{V_4-T_1}}(x_2)&= 1, \\ 
\mu_{\mathbf{f}_{V_1-T_1}\rightarrow x_1}(x_1)&= \delta(x_1,y_{V_1-T_1}), \\   
\mu_{\mathbf{f}_{V_4-T_1}\rightarrow x_1}(x_1)&= \bigvee _{x_2}\delta(x_1+x_2,y_{V_4-T_1}),  \\
\mu_{x_1\rightarrow \mathbf{f}_{V_4-T_1}}(x_1)&= \mu_{\mathbf{f}_{V_1-T_1},x_1}(x_1),  \\
g_{x_1}(x_1)&= \mu_{\mathbf{f}_{V_1-T_1},x_1}(x_1) \; \mu_{\mathbf{f}_{V_4-T_1}\rightarrow x_1}(x_1), \\
\mu_{\mathbf{f}_{V_4-T_1}\rightarrow x_2}(x_2)&= \bigvee _{x_1}\delta(x_1+x_2,y_{V_4-T_1})\; \mu_{\mathbf{f}_{V_1-T_1}\rightarrow x_1}(x_1),  \\ 
g_{x_2}(x_2)&= \mu_{\mathbf{f}_{V_4-T_1},x_2}(x_2),\\
x_1^*&=\mathrm{supt}\;g_{x_1}(x_1),\\
x_2^*&=\mathrm{supt}\;g_{x_2}(x_2).
 \end{align*}
Similar computations apply for $T_2$ also. \hfill $\square$
\end{example}

In the following example we present a network with general demands at sinks and employ the SP algorithm for decoding a vector nonlinear network code for it. We also demonstrate the usefulness of traceback in saving computations of some messages in the factor graph.

\begin{example}
\begin{figure*}[t]
\centering
\includegraphics[scale=0.5]{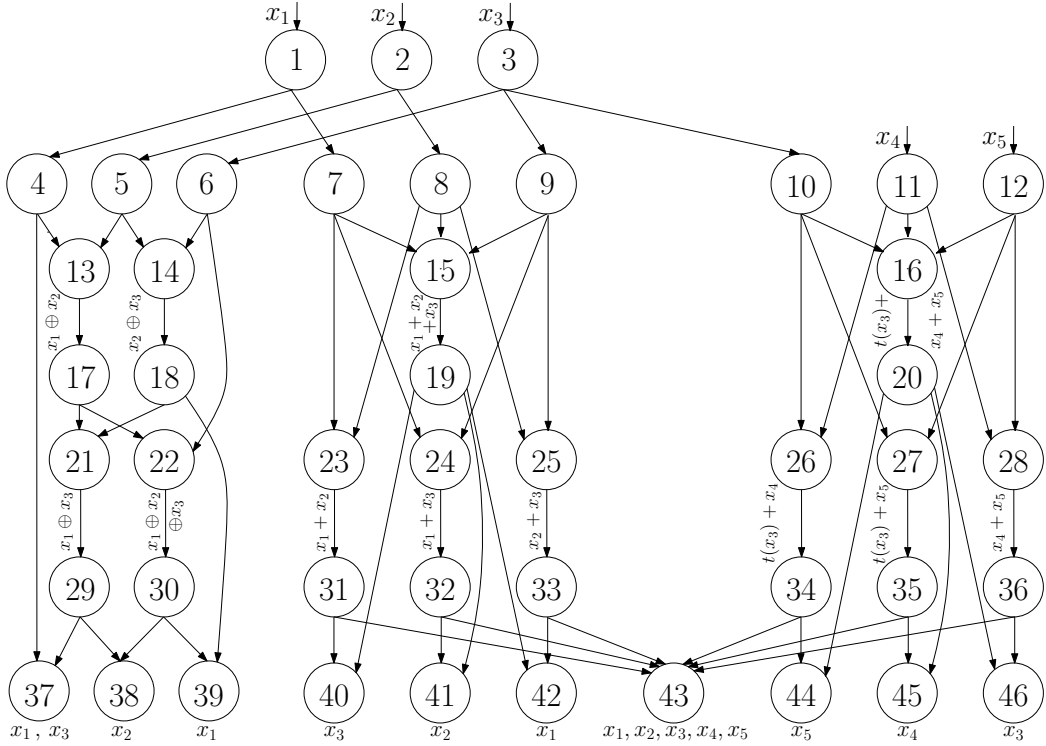}
\caption{The network $\mathcal{N}_3$ of \cite{Insuff}.}
\label{fano}
\end{figure*}
Consider the network given in Fig.~\ref{fano}. The sinks (nodes $37-46$) have general demands which are specified below them. In \cite{Insuff}, the authors showed that this network admits no linear solution over any field and gave a vector nonlinear solution. The source messages $x_i,i\in[5]$ are 2-bit binary words ($q=4,\,\omega = 5$), $+$ denotes addition in ring $\mathbf{Z}_4$, $\oplus$ denotes the bitwise XOR and the function $t(\cdot)$ reverses the order of the 2-bit input. 

The factor graphs for nodes $37$, $40$, and $43$, denoted by $G_{37}$, $G_{40}$, and $G_{43}$ respectively, are given in Fig.~\ref{fano_junc}. The 4-cycle in $G_{40}$ is removed by deleting the dashed edge and stretching variable $x_2$ along the unique path $P$ from $x_2$ to the factor node labeled by $x_1+x_2+x_3$. Similarly, the two 6-cycles in $G_{43}$ are removed by deleting dashed edges and stretching variable $x_3$ along paths $P_1$ and $P_2$; for convenience, nodes are numbered a-k in the acyclic factor graph.
\begin{figure*}[t]
\centering
\includegraphics[scale=0.5]{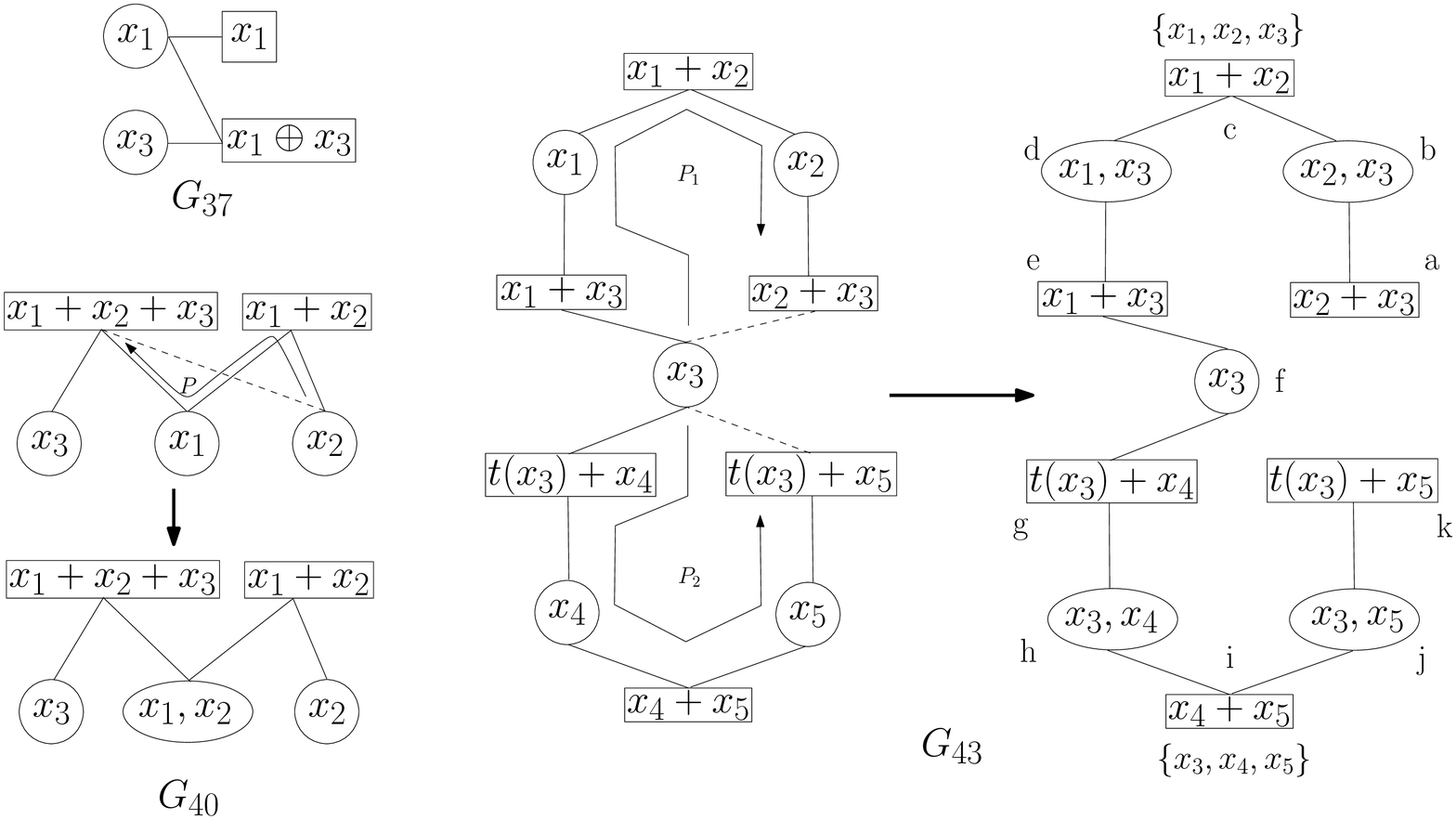}
\caption{The factor graphs for sinks with vertex labels $37$, $40$ and $43$ in network $\mathcal{N}_3$ of Fig.~\ref{fano}. Enlarged local domains after variable stretching are given adjacent to factor nodes.}
\label{fano_junc}
\end{figure*}

We infer from $G_{43}$ that the number of computations required to reproduce all the source messages at $V_{43}$ is only $\mathcal{O}(q^3)$ instead of $\mathcal{O}(q^5)$ (as brute-force decoding would have required). The decoding process at $V_{43}$ is performed by using single-vertex SP algorithm with node ``i'' in $G_{43}$ as the root to compute $x_3$, $x_4$, and $x_5$ followed by traceback to compute $x_1$ and $x_2$. The messages passed towards root are
\begin{align*}
& \mu _{k\rightarrow j}=\mu _{j\rightarrow i}(x_3,x_5)=\delta (t(x_3)+x_5,y_{35-43}),\\
& \mu _{a\rightarrow b}=\mu _{b\rightarrow c}(x_2,x_3)=\delta (x_2+x_3,y_{33-43}),\\
& \mu _{c\rightarrow d}=\mu _{d\rightarrow e}(x_1,x_3)=\bigvee_{\mathclap{x_2}}\delta (x_1+x_2,y_{31-43})\mu _{b\rightarrow c}(x_2,x_3),\\
& \mu _{e\rightarrow f}=\mu _{f\rightarrow g}(x_3)=\bigvee_{x_1}\delta (x_1+x_3,y_{32-43})\mu _{d\rightarrow e}(x_1,x_3),\\
&\text{and}\\
&\mu _{g\rightarrow h}=\mu _{h\rightarrow i}(x_3,x_4)=\delta (t(x_3)+x_4,y_{34-43})\mu _{f\rightarrow g}(x_3).
\end{align*}
Decoding of $x_3$, $x_4$, and $x_5$ is performed at ``i'' by first computing the marginal function $F_{\text{i}}$ using \eqref{eqn_supt} and then computing its support as follows:
\begin{align*}
F_{\text{i}}(x_3,x_4,x_5)&=\delta(x_4+x_5,y_{36-43})\mu _{h\rightarrow i}(x_3,x_4)\mu _{j\rightarrow i}(x_3,x_5),\\
(x_3^*,x_4^*,x_5^*)&=\mathrm{supt}\;F_{\text{i}}(x_3,x_4,x_5).
\end{align*}
Since $x_3^*$ is known, $x_1^*$ and $x_2^*$ are computed using traceback at nodes ``e'' and ``c'' respectively as follows:
\begin{align*}
x_1^*=\mathrm{supt}\;\lambda_e(x_1,x_3^*)
\end{align*}
and
\begin{align*}
x_2^*=\mathrm{supt}\;\lambda_c(x_1^*,x_2,x_3^*),
\end{align*}
where the partial marginal $\lambda_e(x_1,x_3)=\delta (x_1+x_3,y_{32-43})\mu _{d\rightarrow e}(x_1,x_3)$ was computed while passing the message $\mu _{e\rightarrow f}(x_3)$ and $\lambda_c(x_1^*,x_2,x_3^*)=\delta (x_1+x_2,y_{31-43})\mu _{b\rightarrow c}(x_2,x_3)$ was computed while passing the message $\mu _{c\rightarrow d}(x_1,x_3)$. In other words, 
\begin{align*}
\mu _{e\rightarrow f}(x_3)=\bigvee_{x_1}\;\lambda_e(x_1,x_3)
\end{align*}
and
\begin{align*}
\mu _{c\rightarrow d}(x_1,x_3)=\bigvee_{x_2}\;\lambda_c(x_1,x_2,x_3).
\end{align*}
The number of semiring operations required to compute all the messages passed and marginals computed are tabulated in Table~\ref{tab_nops_trcbk}.
\begin{table}[h] 
\caption{Single-vertex SP Algorithm with Traceback}
\label{tab_nops_trcbk}
\centering
\small
\renewcommand{\arraystretch}{1.4}
\begin{tabular}{|c|c|c|c|}
\hline
 & & No. of $\bigwedge$ & No. of $\bigvee$ \\
\hline
$C_1$ & $\mu_{k\rightarrow j}$, $\mu_{j\rightarrow i}$ & $0$ & $0$\\
\hline
$C_2$ & $\mu_{a\rightarrow b}$, $\mu_{b\rightarrow c}$ & $0$ & $0$\\
\hline
$C_3$ & $\mu_{c\rightarrow d}$ & $q^3$ & $q^2(q-1)$\\
\hline
$C_4$ & $\mu_{d\rightarrow e}$ & $0$ & $0$\\
\hline
$C_5$ & $\mu_{e\rightarrow f}$ & $q^2$ & $q(q-1)$\\
\hline
$C_6$ & $\mu_{f\rightarrow g}$ & $0$ & $0$\\
\hline
$C_7$ & $\mu_{g\rightarrow h}$ & $q^2$ & $0$\\
\hline
$C_8$ & $\mu_{h\rightarrow i}$ & $0$ & $0$\\
\hline
$C_9$ & $F_{\text{i}}(x_3,x_4,x_5)$ & $2q^3$ & $0$ \\
\hline
$C_{10}$ & $(x_3^*,x_4^*,x_5^*)$ & $0$ & $q^3-1$ \\
\hline
$C_{11}$ & $x_1^*$, $x_2^*$ & $0$ & $q-1$\\
\hline
\end{tabular}
\end{table}

When not using traceback, SP decoding is performed by computing messages $\mu _{k\rightarrow j}$, $\mu _{j\rightarrow i}$, $\mu _{a\rightarrow b}$, $\mu _{b\rightarrow c}$, $\mu _{c\rightarrow d}$, $\mu _{d\rightarrow e}$, $\mu _{e\rightarrow f}$, $\mu _{f\rightarrow g}$, $\mu _{g\rightarrow h}$, $\mu _{h\rightarrow i}$ as before and then messages $\mu _{i\rightarrow h}$, $\mu _{h\rightarrow g}$, $\mu _{g\rightarrow f}$, $\mu _{f\rightarrow e}$, $\mu _{e\rightarrow d}$, and $\mu _{d\rightarrow c}$ are computed as follows:
\begin{align*}
& \mu _{i\rightarrow h}=\mu _{h\rightarrow g}(x_3,x_4)= \bigvee _{x_5} \delta (x_4+x_5,y_{36-43}) \mu _{j\rightarrow i}(x_3,x_5), \\
& \mu _{g\rightarrow f}=\mu_{f\rightarrow e}(x_3)=\bigvee _{x_4}\delta (t(x_3)+x_4,y_{34-43})\mu _{h\rightarrow g}(x_3,x_4),
\end{align*}
and
\begin{align*}
& \mu_{e\rightarrow d}=\mu_{d\rightarrow c}(x_1,x_3)=\mu_{f\rightarrow e}(x_3)\delta(x_1+x_3,y_{32-43}).
\end{align*}

At ``i'', $x_3^*$, $x_4^*$, and $x_5^*$ are obtained as given above, and $x_1$ and $x_2$ are obtained at ``e'' and ``c'' respectively by first computing the marginal functions and then their appropriate supports using \eqref{eqn_supt_sub} as follows:\\
At ``e''
\begin{align*}
F_{\text{e}}(x_1,x_3)&=\delta(x_1+x_3,y_{32-43})\mu_{f\rightarrow e}(x_3),\\
x_1^*&=\underset{x_1}{\mathrm{supt}}\;F_{\text{e}}(x_1,x_3),
\end{align*}
and at ``c''
\begin{align*}
F_{\text{c}}(x_1,x_2,x_3)&=\delta(x_1+x_2,y_{31-43})\mu_{d\rightarrow c}(x_1,x_3),\\
x_2^*&=\underset{x_2}{\mathrm{supt}}\;F_{\text{c}}(x_1,x_2,x_3).
\end{align*}
The number of semiring operations required to compute additional messages and marginals are tabulated in Table~\ref{tab_nops_notrcbk}.
\begin{table}[h] 
\caption{Multiple-vertex SP Algorithm}
\label{tab_nops_notrcbk}
\centering
\small
\renewcommand{\arraystretch}{1.4}
\begin{tabular}{|c|c|c|c|}
\hline
 & & No. of $\bigwedge$ & No. of $\bigvee$ \\
\hline
$C_{12}$ & $\mu_{i\rightarrow h}$ & $q^3$ & $q^2(q-1)$\\
\hline
$C_{13}$ & $\mu_{h\rightarrow g}$ & $0$ & $0$\\
\hline
$C_{14}$ & $\mu_{g\rightarrow f}$ & $q^2$ & $q(q-1)$\\
\hline
$C_{15}$ & $\mu_{f\rightarrow e}$ & $0$ & $0$\\
\hline
$C_{16}$ & $\mu_{e\rightarrow d}$ & $q^2$ & $0$\\
\hline
$C_{17}$ & $\mu_{d\rightarrow c}$ & $0$ & $0$\\
\hline
$C_{18}$ & $F_{\text{e}}(x_1)$ & $q^2$ & $0$ \\
\hline
$C_{19}$ & $x_1^*$ & $0$ & $q^2-1$ \\
\hline
$C_{20}$ & $F_{\text{c}}(x_2)$ & $q^3$ & $0$ \\
\hline
$C_{21}$ & $x_2^*$ & $0$ & $q^3-1$ \\
\hline
\end{tabular}
\end{table}

Total number of operations (ANDs and ORs) required with traceback is $2C_1+2C_2+C_3+\ldots+C_{10}+2C_{11}$, which is $5q^3+2q^2+q-3=353$ operations, and that without traceback are $2C_1+2C_2+C_3+\ldots+C_{10}+C_{12}+\ldots+C_{21}$, which is $9q^3+6q^2-2q-3=661$ operations. Thus, running single-vertex SP algorithm followed by traceback step affords computational advantage over the multiple-vertex version.\hfill $\square$
\end{example}

\section{Complexity of The SP Algorithm}
We will now determine the number of semiring operations required to compute the desired marginal functions in an MPF problem using the SP algorithm and the desired supports in an arg-MPF problem using the arg-SP algorithm with and without traceback in the Boolean semiring.

In this section, by addition and multiplication we mean the Boolean OR and AND operations. By Remark~\ref{rem_supt_or}, $\mathrm{supt}$ is considered same as addition. Let $G=(Z,E)=(V\cup W,E)$ be an acyclic factor graph with variable nodes $V$ and factor node $W$. The local domain of a node $z$ is denoted by $x_{S_z}$, the cardinality of its configuration space $A_{S_z}$ by $q_z$, and its degree by $d_z$. For an egde $e=(a,b)$ between nodes $a$ and $b$, $q_e=q_{a\cap b}=|A_{S_a\cap S_b}|$ and $q_{a\backslash b}=|A_{S_a\backslash S_b}|$. For every node $z\in Z$, define $a_z=1$ if $z\in W$ and $0$ otherwise. 

\subsection{Single-vertex SP and arg-SP Algorithms}
The message passed from a variable node $v$ to a factor node $w$ as given in \eqref{eqn_msg_mod1} is
\begin{align*} 
\mu _{v\rightarrow w}(x_{S_v\cap S_w})=\quad\bigvee _{\mathclap{x_{S_v\backslash S_w}\in A_{S_v\backslash S_w}}}\qquad\qquad\,\bigwedge _{\mathclap{w'\in N(v)\backslash w}}\,\mu _{w'\rightarrow v} (x_{S_{w'}\cap S_v}).
\end{align*}
In the above equation, for each of the $q_{v\backslash w}$ values of $x_{S_v\backslash S_w}$, product of $d_v-1$ messages is required which requires $d_v-2$ multiplications. For each of the $q_{v\cap w}$ values of $x_{S_v\cap S_w}$, $q_{v\backslash w}-1$ additions and $q_{v\backslash w}(d_v-2)$ multiplications are required. Thus, the total number of operations required are 
\begin{center}
$q_{v\cap w}(q_{v\backslash w}-1)=q_v-q_{v\cap w}$ additions and
$q_{v\cap w}\,q_{v\backslash w}(d_v-2)=q_v(d_v-2)$ multiplications.
\end{center}
The messages passed from a factor node $w$ to a variable node $v$ as given in \eqref{eqn_msg_mod2} is
\begin{align*} 
\mu _{w\rightarrow v}(x_{S_w\cap S_v})=\;\;\bigvee _{\mathclap{x_{S_w\backslash S_v}\in A_{S_w\backslash S_v}}}\,h _w(x_{S_w})\;\;\bigwedge _{\mathclap{v'\in N(w)\backslash v}}\,\mu _{v'\rightarrow w} (x_{S_{v'}\cap S_w}).
\end{align*}
This involves product of a local functions with $d_w-1$ messages for each of the $q_{w\backslash v}$ values of $x_{S_w\backslash S_v}$. The total number of operations required for this case is 
\begin{center}
$q_{w\cap v}(q_{w\backslash v}-1)=q_w-q_{w\cap v}$ additions and
$q_{w\cap v}\,q_{w\backslash v}(d_w-1)=q_w(d_v-1)$ multiplications.
\end{center}
The messages are passed by all nodes except the root node. At the root node $r$, the marginal function is the product of $d_r$ messages, requiring $(d_r-1)q_r$ multiplications, if it is a variable node \eqref{eqn_state_mod1} and the product of $d_r$ messages with the local function, requiring $d_rq_r$ multiplications, if it is a factor node \eqref{eqn_state_mod2}. In other words, computation of marginal function at $r$ requires $(d_r+a_r-1)q_r$ multiplications. Thus, the total number of additions and multiplications required in the single-vertex SP algorithm is
\begin{align*}
\sum_{v\in V\backslash r}(q_v-q_{v\cap w})+\sum_{w\in W\backslash r}(q_w-q_{w\cap v})
=\sum_{z\in Z\backslash r}q_z - \sum_{e\in E}q_e,
\end{align*}
and 
\begin{align*}
\sum_{v\in V\backslash r}q_v(d_v-2)+\sum_{w\in W\backslash r}q_w(d_w-1)+(d_r+a_r-1)q_r\\
=\sum_{z\in Z}(d_z-1)q_z - \sum_{v\in V\backslash r}q_v + a_rq_r.
\end{align*}
The grand total of the number of additions and multiplications is 
\begin{align*}
\mathcal{C}_1=\sum_{z\in Z}d_zq_z - \sum_{e\in E}q_e - \sum_{v\in V}q_v.
\end{align*}
In the arg-SP algorithm, support of marginal at $r$ is computed which requires $q_r-1$ additions (by Remark~\ref{rem_supt_ops}) so that the grand total of operations in this case is 
\begin{align*}
\mathcal{C}_2&=\sum_{z\in Z}d_zq_z - \sum_{e\in E}q_e - \sum_{v\in V}q_v +\, (q_r\, -1)\\
&=\mathcal{C}_1\,+\, (q_r\, -1).
\end{align*}

\subsection{Single-vertex arg-SP Algorithm with Traceback}
In this case, first the single-vertex arg-SP algorithm with $r$ as the root is executed on the factor graph. Then the local domain $x_{S_z}$ of a neighbor $z$ of $r$ is partitioned into sets $x_{I}=x_{S_z\backslash S_r}$ and $x_J=x_{S_z\cap S_r}$. The value $x_J^*$ is already known from decoding at $r$, and $x_I^*$ is computed using \eqref{eqn_trcbk_mc} as follows:
\begin{align*} x^*_I=\mathrm{supt}\; \lambda_z(x_I,x_J^*), \end{align*}
where the table of values of the partial marginal $\lambda_z(x_I,x_J)$ was already computed at $z$ while passing the message $\mu_{z \rightarrow r}(x_J)$ to the root $r$. We need to look only at the rows for which $x_J=x_J^*$ and output the value of $x_I$ for which $\lambda_z(x_I,x_J^*)=1$. This requires $q_I-1<q_z-1$ additions, where $x_I\in A_I$ and $q_I=|A_I|$. The total number of multiplications remains the same as in the single-vertex arg-SP algorithm, which is $\sum_{z\in Z}(d_z-1)q_z - \sum_{v\in V\backslash r}q_v + a_rq_r$, but the number of additions is the sum of the number of additions required in single-vertex SP algorithm $(\sum_{z\in Z\backslash r}q_z - \sum_{e\in E}q_e)$ and the number of additions required at each node, which is at most $\sum_{z\in Z}q_z-1$. Thus, the grand total of operations is at most
\begin{align*}
\mathcal{C}_3&=\mathcal{C}_1 + \sum_{z\in Z}(q_z-1)=\mathcal{C}_2 + \sum_{z\in Z\backslash r}(q_z-1)\\
&=\sum\limits_{z\in Z}d_zq_z - \sum\limits_{e\in E}q_e + \sum\limits_{v\in V}q_v \,+\, \sum_{z\in Z}(q_z-1)\\
&=\sum\limits_{z\in Z}d_zq_z - \sum\limits_{e\in E}q_e + \sum\limits_{w\in W}q_w \,-\, |Z|.
\end{align*}

\subsection{All-vertex SP and arg-SP Algorithms}
In the all-vertex SP algorithm, first the messages are passed by all the nodes on the unique path towards the root. When the root has received messages from all its neighbors, messages are passed on each edge in the reverse direction, i.e., away from the root and towards the leaves. When all the leaves have received the messages, marginal functions of each node is computed. We use the method suggested in \cite[Sec.~V]{Aji} to compute messages and marginal function.

Let a node $z$ have degree $d$ and has received messages from all but one of its neighbors $z'$ which is on the unique path from $z$ to the root. For an instance $x_{S_z}'$ of $x_{S_z}$, let $k_2,k_3,\ldots,k_d$ be the values of the known messages, $k_1$ be the value of the message it is yet to receive from $y_1$, and $h_z$ be the value of its local function, assumed to be $1$ if $z\in V$, i.e., $k_i=\mu_{y_i\rightarrow z}(x_{S_z}')$, $y_i\in N(z)$. The messages $\hat{k}_i$ involves the product of $h_z$ with all $k_j$s excluding $k_i$ and summing over suitable variables as in \eqref{eqn_msg_mod1} and \eqref{eqn_msg_mod2}; there are $d$ such messages to be sent, one to each neighbor.
\begin{figure}[h]
\centering
\includegraphics[scale=0.5]{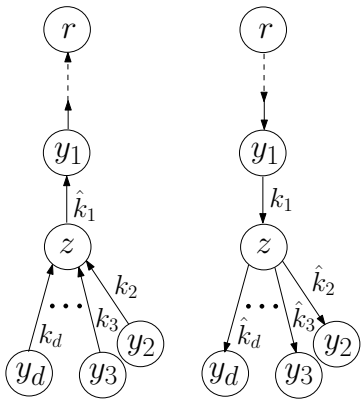}
\caption{Message passing at $z$.}
\label{cplxty}
\end{figure}
This can be achieved by computing the following products consecutively: $c_d=h_zk_d$, $c_{d-1}=c_dk_{d-1}=h_zk_{d-1}k_d$, $\ldots$, $c_3=c_4k_3=h_zk_3k_4\ldots k_d$, $c_2=c_3k_2=h_zk_2k_3\ldots k_d$; this step requires $d-2+a_z$ multiplications. Now $z$ passes $\hat{k}_1=c_2$ to $y_1$ (after summing over suitable variables) and awaits the reception of $k_1$ from $z'$. Once $k_1$ is received, the marginal functions is computed, $g_z=k_1c_2=h_zk_1k_2\ldots k_d$, which requires $1$ multiplication. Then the following products are computed consecutively: $b_1=k_1$, $b_2=b_1k_2=k_1k_2$, $b_3=b_2k_3=k_1k_2k_3$, $\ldots$, $b_{d-1}=b_{d-2}k_{d-1}=k_1k_2\ldots k_{d-1}$; this step requires $d-2$ multiplications. Subsequently, $\hat{k}_i$s are computed as follows: $\hat{k}_2=b_1c_3=h_zk_1k_3k_4\ldots k_d$, $\hat{k}_3=b_2c_4=h_zk_1k_2k_4\ldots k_d$, $\ldots$, $\hat{k}_{d-1}=b_{d-2}c_d=h_zk_1k_2\ldots k_{d-2}k_d$, $\hat{k}_d=h_zb_{d-1}=h_zk_1k_2\ldots k_{d-1}$; this step requires $d-2+a_z$ multiplication. Various messages received and passed by node $z$ are depicted Fig.~\ref{cplxty}. Thus, computation of all the messages to be passed by $z$ and its marginal function requires $(d-2+a_z)+1+(d-2)+(d-2+a_z)=3(d-2)+2a_z+1$ multiplications for each of the $q_z$ values in $A_{S_z}$. This is true for the root node also. Hence, total number of multiplications required is $\sum_{z\in Z} q_z\,[3(d_z-2) \,+\, 2a_z \,+\, 1]$. The number of additions required for computing each message remains the same as in the single-vertex SP algorithm, $q_v-q_{v\cap w}$ for a variable node passing message to $w$. Unlike in the single-vertex case, now $v$ will pass messages to all its $d_v$ neighbors, thus requiring $d_vq_v-\sum_{e \text{ incident on }v}q_e$. Same is true for all the factor nodes also. Hence, the total number of additions required is $\sum_{z\in Z}d_zq_z \,-\, 2\sum_{e\in E}q_e$. The grand total number of operations is then
\begin{align*}
\mathcal{C}_4 &= \sum_{\mathclap{z\in Z}} q_z(3d_z-5 + 2a_z) + \sum_{\mathclap{z\in Z}}d_zq_z - 2\sum_{\mathclap{e\in E}}q_e
\\ &=\sum_{\mathclap{z\in Z}}(4d_z-5)q_z + 2\sum_{w\in W}q_w - 2\sum_{\mathclap{e\in E}}q_E.
\end{align*}

In the arg-SP algorithm, computation of support of marginal function at a node $z$ requires at most $q_z-1$ additions. Thus, the total number of operations required in all-vertex arg-SP algorithm is 
\begin{align*}
\mathcal{C}_5=\mathcal{C}_4\,+\, \sum_{z\in Z}(q_z-1)=\mathcal{C}_4\,+\, \sum_{z\in Z}q_z-|Z|.
\end{align*}

The results of Sections~IV-A,B, and C are tabulated in Table~\ref{tab_nops}. The operation counts presented in this section apply not only to MPF and arg-MPF problem in Boolean semiring, but also to MPF and arg-MPF problem in min-sum, min-product, max-sum, and max-product semiring.

\subsection{Utility and Complexity of SP Algorithm for Decoding Network Code}
The SP algorithm for decoding a network code is advantageous when the code is either nonlinear or it is linear but the number of messages is very large. For linear network codes with manageable value of $\omega$, Gaussian elimination with backward substitution is advisable. 

For a node that demands all the source messages, for example a sink in a multicast network, if application of SP algorithm for decoding network codes leads to computational complexity strictly better than the brute-force decoding complexity, then the code is called a \emph{fast SP decodable network code}. The network code for network $\mathcal{N}_3$ given in Fig.~\ref{fano} is fast SP decodable for the sink with vertex label $43$; decoding complexity is only $\mathcal{O}(q^3)$ compared to the brute-force complexity of $\mathcal{O}(q^5)$.

As stated above, in order to recover the requisite source messages at a sink we need only run the single-vertex arg-SP algorithm followed by traceback steps. For a given sink node, if the factor graph constructed using the method given in Section~III-A is cycle-free and the network code is such that the local domains of all factor nodes have cardinality at most $l\, (<\omega)$, then the number of operations required for decoding using the SP algorithm is $\mathcal{O}(q^l)$. If the sink demands all the source messages, then the brute-force decoding would require $\mathcal{O}(q^{\omega}) (> \mathcal{O}(q^l))$ operations. Thus, an acyclic factor graph with at most $l\, (<\omega)$ variables per equation is a sufficient condition for fast decodability of the network code at a sink which demands all the source messages.

If the graph is not cycle-free then we remove the cycles by variable stretching. Let $m\leqslant \omega$ be the size of maximum cardinality local domain in the new cycle-free factor graph. The number of computations required now will be $\mathcal{O}(q^m)\leqslant \mathcal{O}(q^{\omega})$ and the code is fast decodable iff $m<\omega$.

\section{In-network Function Computation Using The SP Algorithm}
\subsection{Preliminaries}
In a communication network, some nodes may be interested not in the messages generated by some other nodes but in one or more functions of messages generated by other nodes. For example, in a wireless sensor network that comprises several sensor nodes, each measuring environmental parameters like ambient light, temperature, pressure, humidity, wind velocity etc. For long-term record-keeping and weather forecasting, average, minimum, maximum and variance of these meteorological parameters are of interest. Environmental monitoring in an industrial unit is another field of application where relevant parameter may include temperature and level of exhaust gases which may assist in preventing fire and poisoning due to toxic gases respectively. 

We consider in-network function computation in a finite directed acyclic error-free network, $\mathcal{G=(V,E)}$, where codes can perform network coding. For brevity of expression, we use $x$ for $x_{[\omega]}$ in this section. The network model is same as given in Section~I-A for network coding problem with the exception that the sink nodes demand a function of messages rather than a subset of messages, i.e., a sink node $T_k$ demands the function $g_k:F^{\omega}\rightarrow F$. A network code comprises global encoding maps $\tilde{f}_e:F^{\omega}\rightarrow F$, one for each edge $e\in E$, such that there exist $K$ (decoding) maps, $\mathcal{D}_k:F^{|In(T_k)|}\rightarrow F$, for each sink $T_k,\,k\in [K]$, such that $\mathcal{D}_k(y_e:e\in In(T_k))=g_k(x)$. This subsumes the network coding problem of Section~I as a special case. By $(y_e:e\in In(T_k))$ we denote the $|In(T_k)|$-tuple of coded messages received by $T_k$ on its incoming edges.
\begin{remark} \label{rem_func}
Though arguments of a demanded function $g$ may only be a subset, say $\{x_I\}$ for some $I\subseteq [K]$, of messages, we assume it to be a map from $F^{\omega}$ to $F$ for simplicity rather than from $F^{|I|}$ to $F$.
\end{remark}
\begin{remark}
If a sink demands $N\,(>1)$ functions, then such a sink may be replaced by $N$ sinks each demanding one function but the incoming information to these new sinks is the same (see Fig.~\ref{multi_func}). 
\begin{figure}[h]
\centering
\includegraphics[scale=0.7]{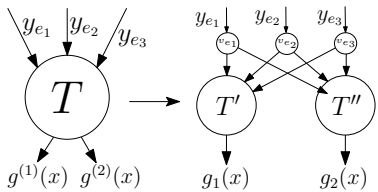}
\caption{Converting a sink that demands multiple functions into multiple sinks each with single demand.}
\label{multi_func}
\end{figure}
\end{remark}

The in-network function computation problem is to design network code that maximizes the frequency of target functions computation, called the \emph{computing capacity}, per network use. In \cite{GiriKr}, bounds on rate of computing symmetric functions (invariant to argument permutations), like minimum, maximum, mean, median and mode, of data collected by sensors in a wireless sensor network at a sink node were presented. The notion of min-cut bound for the network coding problem \cite{Yeung} was extended to function computation problem in a directed acyclic network with multiple sources and one sink in \cite{Appu}. The case of directed acyclic network with multiple sources, multiple sinks and each sink demanding the sum of source messages was studied in \cite{Dey}; such a network is called a sum-network. Relation between linear solvability of multiple-unicast networks and sum-networks was established. Furthermore, insufficiency of scalar and vector linear network codes to achieve computing capacity for sum-networks was shown. Coding schemes for computation of arbitrary functions in directed acyclic network with multiple sources, multiple sinks and each sink demanding a function of source messages were presented in \cite{Dey2}. In \cite{Appu2}, routing capacity, linear coding capacity and nonlinear coding capacity for function computation in a multiple source single sink directed acyclic network were compared and depending upon the demanded functions and alphabet (field or ring), advantage of linear network coding over routing and nonlinear network coding over linear network coding was shown.

In order to obtain the value of its desired functions, a sink node may require to perform some operations on the messages it receives on the incoming edges. Though there are many results on bounds on the computing capacity and coding schemes for in-function computation problem, the decoding operation to be performed at the sink nodes to obtain the value of the desired functions has not been studied. We now formulate computation of the desired functions at sink nodes as an MPF problem over the Boolean semiring and use the SP algorithm on a suitably constructed factor graph for each sink to obtain the value of the desired functions.

\subsection{Function Computation as an MPF Problem}
We consider decoding at the sink node $T_k$. It demands the function $g_k(x_{I_k})$, where $\{x_{I_k}\}=\{x_{i_1},x_{i_2},\ldots,x_{i_{|I_k|}}\}$ is the set of arguments of $g_k$ for some $I_k\subseteq [K]$. For a realization $x^*$ of the message vector, we are interested in the value $G_k^*=g_k(x_{I_k}^*)$. Since a network code only ensures computation of the correct value $G_k^*$ of the demanded target function given the incoming coded message vector $(y_e:e\in In(T_k))$ and not the realization $x_{I_k}^*$ of the messages in the argument set, there may be multiple $|I_k|$-tuples that produce the same values of the incoming coded messages and function value when input to the demanded function, i.e., the network code is a many-to-one mapping. We denote one such message vector by $\widehat{x}_{I_k}$. It need not necessarily be equal to $x_{I_k}^*$ but $\tilde{f}_e(\widehat{x}_{I_k})=\tilde{f}_e(x^*_{I_k})$ for all $e\in In(T_k)$ and $g_k(\widehat{x}_{I_k})=g_k(x_{I_k}^*)$. Using the SP algorithm, we will first obtain $\widehat{x}_{I_k}$ and then evaluate $g(\widehat{x}_{I_k})$ to obtain $G_k^*$. The arg-MPF formulation for obtaining $\widehat{x}_{I_k}$ is given below. Let
\begin{align}
\label{eqn_mpf_func}
S_k= \underset{I_k}{\mathrm{supt}}\;\;\bigwedge _{\mathclap{e \in In(T_k)}} \delta \left( \tilde{f}_e (x)\, , \, y_e \right)=\underset{I_k}{\mathrm{supt}}\;{\beta ^{(k)}(x)} 
\end{align}
Here $\beta ^{(k)}$ is the global product function and $\delta(\tilde{f}_e(x),\,y_e)$ are the local functions of the MPF problem at the sink $T_k$. The set $S_k$ contains the coordinates indexed by $I_k$ of the message vectors $x$ for which $\beta ^{(k)}(x)=1$, i.e., the coordinates indexed by $I_k$ of all those message vectors for which $\tilde{f}_e(x)=y_e$, for all $e\in In(T_k)$. Though $\mathrm{supt}_{I_k}$ may output multiple $|I_k|$-tuples, we will choose any one as $\widehat{x}_{I_k}$. The desired function values is then
\begin{align}
\label{eqn_func_eval}
G_k^*=g_k(\widehat{x}_{I_k})
\end{align}
Thus, the function computation can be performed by using SP algorithm to solve MPF problem in \eqref{eqn_mpf_func} followed by \eqref{eqn_func_eval}.

\begin{theorem}
For all $s\in S_k$ obtained using \eqref{eqn_mpf_func} and each $k\in [K]$, we have $g_k(s) = g_k(x_{I_k}^*)$.
\end{theorem}

\begin{proof}
By Remark~\ref{rem_func}, $g_k(x_{I_k})=g_k(x)$. A look-up table (LUT) approach to decoding is to maintain a table with $q^{\omega}$ rows and two columns at each sink: first column containing all possible incoming message vectors, $\{(\tilde{f}_e(x):e\in In(T_k)):x\in F^{\omega}\}$, and the second column listing corresponding values of the demanded function, $\{g_k(x):x\in F^{\omega}\}$. Given an instance of incoming messages, a sink node locates the row containing that $|In(T_k)|$-tuple in the first column of the LUT and then outputs the value in the second column of the row, which is the desired function value. If two rows in the LUT have the same entry in the first column (network code is a many-to-one map), the entry in the second column will also be same. On the contrary, if for two $x\neq x'$, $g_k(x)\neq g_k(x')$ but $\tilde{f}_e(x)= \tilde{f}_e(x')$ for all $e\in In(T_k)$ and some $k\in [K]$, then there will be ambiguity at the $k$th receiver because there are two distinct possible function values, $g_k(x)$ and $g_k(x')$, that the decoder may output. 

Thus, a valid network code that fulfills all receivers' demands satisfies $\tilde{f}_e(x)\neq \tilde{f}_e(x')$ for all $e\in In(T_k)$ if $g_k(x)\neq g_k(x')$ for each $k\in [K]$ and $x\neq x'$, $x,x'\in F^{\omega}$. 

Let $x^*$ be a realization of the message vector and $(y_e:e\in In(T_k))$ the coded message received by $T_k$ on its incoming edges. The set 
\begin{align*}
S'_k&= \mathrm{supt}  {\bigwedge _{e \in In(T_k)} \delta \left( \tilde{f}_e (x)\, , \, y_e \right)}
\end{align*}
contains all the message vectors $s'\in F^{\omega}$ such that $(\tilde{f}_e (s'):e\in In(T_k))=(y_e:e\in In(T_k))$ including $x^*$. Thus, $g_k(s')=g_k(x^*)$ for all $s'\in S'_k$. Since $S_k=\{s'_{I_k}:s'\in S'_k\}$ and $g_k(x_{I_k})=g_k(x)$, we have that $g_k(s)=g_k(x_{I_k}^*)$ for all $s\in S_k$.
\end{proof}

Hence, the SP algorithm for \eqref{eqn_mpf_func} can terminate as soon as a message vector $\widehat{x}_{I_k}$ with $\beta ^{(k)}(\widehat{x}_{I_k})=1$ is found and we need not obtain all possible message vectors which evaluate to the given coded messages on incoming edges of a sink.

\begin{example}
For example, let $\omega=4$, $x_i\in \mathbb{F}_2$ for all $i\in [4]$, and $g(x_1,x_2,x_3)=x_1+x_2+x_3+ Maj(x_1,x_2,x_3)$ needs to be evaluated using $\tilde{f}_{e_1}=x_1+x_2,\,\tilde{f}_{e_2}=x_2+x_3$, and $\tilde{f}_{e_3}=x_1+x_3$. Here $I=\{1,2,3\}$. Let $x^*=1110$ be a realization of the message vector. Then, $y_{e_1}=0$, $y_{e_2}=0$, $y_{e_3}=0$, and $g(x^*)=0$. From \eqref{eqn_mpf_func}, we have
\begin{align*}
S\;=\;\underset{I}{\mathrm{supt}}\bigwedge_{j\in [3]}\delta(\tilde{f}_{e_j}(x),y_{e_j}) \;=\; \{000,111\}
\end{align*}
Any element of $S$ can be chosen as $\widehat{x}_I$ and both evaluate to $0$ when input to $g(x_I)$. This illustrates that $g(\widehat{x}_I)=g(x_I^*)$. \hfill $\square$
\end{example}

The factor graph for computation of function $g_k(x_{I_k})$ at sink $T_k,k\in [K]$ is constructed as follows:
\begin{enumerate}
\item Install $\omega$ variable nodes, one for each source message. These vertices are labeled by their corresponding source messages, $x_i$.
\item Install $|In(T_k)|$ factor nodes and label them $\tilde{f}_e, e\in In(T_k)$. The associated local domain of each such vertex is the set of source messages that participate in that encoding map and the local kernel is $\delta(\tilde{f}_e(x)\, , \, y_e)$.
\item A variable node is connected to a factor node iff the source message corresponding to that variable node participates in the encoding map corresponding to the said factor node.
\item Install an additional dummy factor node with local domain $\{x_{I_k}\}$, local kernel $1$ and label it $g_k$. Connect this node to variable nodes in the set $\{x_{I_k}\}$, i.e., to the arguments of $g_k$. This node corresponds to the demanded function.
\end{enumerate}

As before, first the cycles in the factor graph are removed, if there are any. The single-vertex SP algorithm is run on the acyclic factor graph with the dummy factor node as the root using \eqref{eqn_msg_mod1} and \eqref{eqn_msg_mod2}. Once it has received all the messages, its marginal function (using \eqref{eqn_state_mod2}) and subsequently the set $S_k$ are computed as follows:
\begin{align*}
S_k=\mathrm{supt}\bigwedge _{v\in N(g_k)}\mu _{v\rightarrow g_k}(x_{S_v\cap I_k}),
\end{align*}
where $S_v$ is the local domain of a neighboring variable node $v$ of $g_k$. Theorem~1 states that obtaining only an element $\widehat{x}_{I_k}$ of the set $S_k$ is sufficient to get the desired function value $G_k^*=g_k(\widehat{x}_{I_k})$.

\section{Discussion}
In this paper, we proposed to use the SP algorithm for decoding network codes and performing in-network function computation. We posed the problem of network code decoding at each sink node in a network as an MPF problem over the Boolean semiring. A method for constructing a factor graph for a given sink node using the global encoding maps (or vectors in case of an LNC) of the incoming edges and demands of the sink was provided. The graph so constructed had fewer nodes and led to fewer message being passed lowering the number of operations as compared to the scheme of \cite{Salmond}. We discussed the advantages of traceback over multiple-vertex SP algorithm. The number of semiring operations required to perform the SP algorithm with and without traceback were derived. For the sinks demanding all the source messages, we introduced the concept of fast decodable network codes and provided a sufficient condition for a network code to be fast decodable. Then we posed the problem of function computation at sink nodes in an in-network function computation problem as an MPF problem and provided a method to construct a factor graph for each sink node on which SP algorithm can be run to solve the MPF problem.

Using the SP algorithm to decode network error correcting codes is a possible direction of future work.


\end{document}